\documentclass[submission,copyright,creativecommons]{eptcs}
\providecommand{\event}{SR 2020} 
\usepackage{latexsym,amssymb,theorem,amsmath,graphicx}

\newcommand{\dom}{\mathrm{dom}}

\newcommand{\Ax}{\mathit{Ax}}

\newcommand{\LTL}{\mathrm{LTL}}
\newcommand{\CTL}{\mathrm{CTL}}
\newcommand{\QCTL}{\mathrm{QCTL}}
\newcommand{\PLTL}{\mathrm{PLTL}}
\newcommand{\ATL}{{\rm ATL}}

\newcommand{\infin}{\mathrm{inf}}
\newcommand{\fin}{\mathrm{fin}}
\newcommand{\Subf}{\mathrm{Subf}}

\newcommand{\SL}{\mathrm{SL}}

\newcommand{\sem}[1]{[\![#1]\!]}
\newcommand{\atlD}[1]{\langle\!\langle#1\rangle\!\rangle}
\newcommand{\angles}[1]{\langle#1\rangle}
\newcommand{\atlB}[1]{[\![#1]\!]}

\newcommand{\past}{{\Diamond\hspace{-0.8em}\mbox{\raisebox{0.2ex}[1.4ex]{$-$}}}}
\newcommand{\necpast}{\mathbf{\boxminus}}
\newcommand{\yesterday}{\mathbf{\ominus}}

\def\Next{{\mathop{\mbox{\raisebox{1.7pt}{$\scriptstyle\bigcirc$}}}}}
\newcommand{\until}[2]{{({#1}{\mathsf U}{#2})}}

\newcommand{\since}[2]{{({#1}{\mathsf S}{#2})}}

\newcommand{\I}{{\mathsf{I}}}

\newcommand{\Act}{{\mathit{Act}}}

\newcommand{\Var}{{\mathrm{Var}}}
\newcommand{\Cl}{{\mathit{Cl}}}

\newcommand{\Ag}{{\Sigma}}

\newcommand{\outcome}{{\mathsf{out}}}

\def\defeq{{\,\hat{=}\,}}

\newcommand{\parag}[1]{\vspace{2mm}\noindent{\bf #1}\ }

\def\crat{{\widetilde{rat}}}
\newcommand{\partition}[1]{{\mathsf{part}(#1)}}

\newcommand{\bm}[1]{\mathbf{#1}}

\newcommand{\oomit}[1]{}

\newtheorem{all}{Proposition}
\newtheorem{lemma}[all]{Lemma}

\newtheorem{theorem}[all]{Theorem}

{\theorembodyfont{\upshape}\newtheorem{definition}[all]{Definition}}
{\theorembodyfont{\upshape}\newtheorem{example}[all]{Example}}
\newenvironment{proof}{\noindent
{\bf Proof:}}{$\dashv$}
\newenvironment{proof*}{\noindent
{\bf Proof:}}{$\dashv$}

\def\maxval{\mathsf{optimal}}

\def\doubleNotLess{{\,\hspace{0.2em}\not\hspace{-0.2em}\nless\,}}
\newcommand{\SPTC}{\mathrm{SPTC}}

\title{Reasoning about Temporary Coalitions and $\LTL$-definable\\
Ordered Objectives in Infinite Concurrent Multiplayer Games\\
{\small Extended Abstract}}
\author{Dimitar P. Guelev
\institute{Institute of Mathematics and Informatics\\
Bulgarian Academy of Sciences, Sofia, Bulgaria}
\email{gelevdp@math.bas.bg}
}
\def\titlerunning{Temporary Coalitions and Ordered Objectives in Infinite Concurrent Multiplayer Games}
\def\authorrunning{D. P. Guelev}
\begin{document}
\maketitle

\begin{abstract}
We propose enhancing the use of propositions for denoting decisions and strategies as established in temporal languages such as $\CTL^*$, if interpreted on concurrent game models. The enhancement enables specifying varying coalition structure. In quantified $\CTL^*$ this technique also enables quantifying over coalition structure, and we use it to quantify over an extended form of strategy profiles which capture temporary coalitions. We also extend $\CTL^*$ by a temporal form of a binary preference operator that can be traced back to the work of Von Wright. The resulting extension of quantified $\CTL^*$ can be used to spell out conditions on the rationality of behaviour in concurrent multiplayer games such as what appear in solution concepts, with players having multiple individual objectives and preferences on them, and with the possibility to form temporary coalitions taken in account. We propose complete axiomatisations for the extension of $\CTL^*$ by the temporal preference operator. The decidability of the logic is not affected by that extension. 
\end{abstract}

\noindent{\bf keywords:} strategic ability, temporary coalitions, rational synthesis, preference, concurrent multiplayer games.

\section*{Introduction}

In Alternating-time Temporal Logic ($\ATL$, \cite{AHK97,AHK02}), $\atlD{\Gamma}\varphi$ means that a joint strategy for the players from $\Gamma$ exists which, if implemented, would cause the (rest of the) play to satisfy temporal condition $\varphi$. Willfully carrying out a strategy to enforce a temporal property which refers to the entire future requires $\Gamma$ to act as a {\em permanent} coalition with that property as the {\em shared} objective and the joint strategy as the member players' lifetime agenda. This means that $\ATL$'s game-theoretic operator, if regarded as a pattern for statements, is an immediate match for statements on the strategic ability of permanent coalitions wrt shared long term objectives. Strategy Logic ($\SL$, \cite{DBLP:conf/fsttcs/MogaveroMV10}) is built around a construct which expresses that implementing a given global strategy profile would enforce a temporal property that similarly refers to the entire future of the play. This is not to say that permanent coalitions are the limit of expressive power for $\ATL$ and $\SL$ in some ultimate sense. The skillful use of logical connectives, superposition and quantification can certainly be relied on to achieve a variety of other statements about strategic ability. 

In this paper we consider a setting where Concurrent Game Models (CGMs) and the temporal sublanguage take their established roles much like in $\ATL$, but aiming to formulate statements about strategic ability and rational behaviour on behalf of players who are prepared to enter {\em temporary} coalitions while in pursuit of their {\em individual} long term objectives, without having to combine those objectives into long term common ones such as what the plain use of the constructs offered in $\ATL$ and $\SL$ appears to be best suited for. To this end we fall back on propositionally quantified $\CTL^*$ (henceforth $\QCTL^*$) as the temporal logic language while keeping CGMs as the semantics. We extend the established way of using propositional variables for naming strategies and expressing quantification over strategies in $\QCTL^*$ to provide notation for speculating about the varying coalition structure and the 'short term' agendas of temporary coalitions combined.

To enable judgements which link the viability of temporary coalitions and their agendas to the players' individual long term objectives, we assume that the preferences of each player define a partial order on plays, and extend $\CTL^*$ by a temporal binary preference operator: $\varphi<_i\psi$ means that $i$ prefers plays which satisfy $\varphi$ to plays satisfying $\psi$. The operator can be viewed a temporal extension of the abstract binary preference operator which can be traced back to the work of Von Wright \cite{vonWrightPreference}. Interestingly, the possibility to specify step-by-step progress towards objectives, and the evolution of objectives in the course of their partial fulfillment or forfeiture admit a rather straightforward expression using this operator as an addition to $\CTL^*$ in combinations with $\LTL$'s guarded normal form and the separated normal form of $\LTL$ formulas. The interaction is encoded in some axioms for $<_i$ we propose as part of the second of two complete axiomatisations for it. The first axiomatisation enables the elimination of $<_i$, in case the preferences of every player are suitably expressed by some finite list of comparisons between $\LTL$-definable sets of plays. The second axiomatisation makes no such assumption; indeed its completeness proof, which is not included in this extended abstract, entails that a formula in $\CTL^*$ with $<_i$ is satisfiable only if so wrt the preferences of the players being expressible this way. Since the possibility to work towards objectives in multiple steps is specific to those objectives being temporal, we believe that this temporal species of a binary preference operator is of some interest on its own too. 

Meta-theoretical results on {\em temporalization} were first obtained in \cite{DBLP:journals/jolli/FingerG92}. More general ways for combining logics were later proposed by Gabbay and others. Our extension of $\ATL^*$ with $<$ is broadly compatible with that framework. In the sense of \cite{DBLP:journals/igpl/GabbayS98}, it can be viewed as a kind of {\em fusion} of $\ATL^*$ and Preference Logic as in \cite{van2005preference}, except that the state space we interpret Preference Logic in CGMs on is not the same as that for $\ATL^*$ but is the set of the infinite plays in the considered CGM. 

Finding strategy profiles which meet the requirements of solution concepts in general is referred to as {\em rational synthesis}. For {\em temporal} objectives, rational synthesis was proposed in \cite{DBLP:conf/tacas/FismanKL10} and further investigated in \cite{DBLP:journals/amai/KupfermanPV16} with the focus on rational behaviour in the sense of Nash equilibrium wrt the objectives of some finite number of players, with the additional condition that the objective of a distinguished player, the {\em system}, is achieved. One key observation about solution concepts with coalitions is that players can use their knowledge of the actions to be taken by their allies when deciding whether to enter a coalition and embrace its agenda. Non-members are just assumed to be acting rationally wrt their respective preferences and the adopted solution concept. With temporary coalitions, this includes perpetually revising the coalition structure. In this paper, we formulate temporary coalition variants of the concepts of the {\em game core} and {\em dominant strategies} to demonstrate the use of the  notation we are proposing. We do not claim these variants to be the unique options for extending these solution concepts. Indeed, solution concepts being so varied is what justifies studying logical languages which are both versatile and amenable to algorithmic methods. Results on rational synthesis after \cite{DBLP:conf/tacas/FismanKL10} in the literature are now commonly derived just from the possibility to express the solution concepts in a decidable subset of $\SL$ too. Of course, optimality requires looking at algorithms which are taylored for specific concepts; cf. e.g., \cite{DBLP:journals/corr/BouyerBMU15,DBLP:conf/fossacs/BouyerBMU12} on Nash equilibrium. 

\oomit{
We do not consider a form of {\em Nash equilibrium} for temporary coalitions; for this concept to work, advance agreement on the equilibrium to be implemented appears to be necessary, and this implies a grand permanent coalition of sorts to be in place, with the solution concept being about its resilience. This renders the meaningfulness of temporary coalitions questionable. The study of {\em $k$-resilience} \cite{DBLP:conf/gamesec/Halpern11}, which is about group defections from Nash equilibria, can count as mildly touching on the possible interaction of that concept with temporary coalitions. 
} 

\oomit{
In reality temporary coalitions require negotiation to precede joint action. Interleaving negotiation steps and representing negotiation in terms of individual players' actions is essential for deterministic individual strategies to be possible to use. However, negotiation affects mostly  {\em mental} states and is often confidential. This calls for an incomplete information setting, well beyond the scope of this paper. With complete information, we cannot but go for solution concepts which abstract negotiation away and only assume that a rational partitioning of the players into coalitions, possibly a different one at each step, is always achieved with rationality according to some solution concept. The rationality of participating in a temporary coalition depends on the gains upon dissolving the coalition. This makes {\em backward induction} very intuitive for temporary coalitions. In our axioms about {\em backward induction} we allow for plays to be partly terminating and partly non-terminating, because of the multiplicity of the players objectives, and conclusiveness of what can be derived using the axioms can be similarly partial.
}

\oomit{

Our modelling incorporates a binary preference operator $\varphi<\psi$, {\em a play satisfying $\varphi$ is preferable to one satisfying $\psi$}, which can be traced back to \cite{vonWrightPreference} and, in various non-temporal forms has been studied in more recent works, cf. e.g. \cite{DBLP:conf/kramas/DegremontK08,Lorini2010}. Our main tool is an established model transformation. In logics for agency, it was used to demonstrate the equivalence between two semantics of $\ATL$ in \cite{GJ04} and later, for other purposes, in our works \cite{DBLP:conf/clima/GuelevD12} and \cite{DBLP:journals/corr/abs-1303-0794}. It is similar to the transformation of finite automata with an input alphabet into so-called {\em state-based} automata, where input letters label the target states of transitions, and not the transitions themselves. The transformation facilitates the naming of strategy profiles in the temporal language, and the lack of names for strategies is the main disadvantage of $\ATL$s in comparison with $\SL$s. Temporal languages have been extended to allow naming decisions and memoryless strategies in other ways in \cite{DBLP:conf/atva/Pinchinat07,DBLP:conf/concur/LopesLM12} too. 

}

\section{Preliminaries}
\parag{Concurrent Game Models} 
We consider CGMs $M=\angles{W,w_I,\angles{ \Act_i:i\in\Ag},o,V}$ for some given sets of {\em players} $\Ag=\{1,\ldots,N\}$ and atomic propositions $AP$. $W$ is the set of {\em states}, $w_I\in W$ being the {\em initial} state, $\Act_i$ is the set of {\em actions} of player $i\in\Ag$, $\Act_\Gamma\defeq\prod\limits_{i\in\Gamma}\Act_i$, $o:W\times\Act_\Ag\rightarrow W$ is the {\em outcome} function, and $V\subseteq W\times AP$ is the {\em valuation} relation. Given a $\bm{w}\in W^+$, we put 
\[R^\infin_{M}(\bm{w})\defeq\{\bm{v}\in W^\omega:\bm{v}^0\ldots\bm{v}^{|\bm{w}|-1}=\bm{w},(\forall k<\omega)(\exists a\in\Act_\Ag)(\bm{v}^{k+1}=o(\bm{v}^k,a))\}.\]
for the set of all the infinite plays in $M$ which are continuations of $\bm{w}$. The set $R^\fin_{M}(\bm{w})\subseteq W^+$ of the {\em finite} continuations of $\bm{w}$ is defined similarly. 

Given a $p\in AP$ and an $X\subseteq W$, $M^X_p$ denotes the CGM $\angles{W,w_I,\angles{ \Act_i:i\in\Ag},o,V^X_p}$ where $V^X_p(w,p)\leftrightarrow w\in X$ and $V^X_q(w,q)\leftrightarrow V(w,q)$ for $q\in AP\setminus\{p\}$.

In the sequel we tacitly assume an arbitrary fixed finite CGM $M$ with its components named as above.

\parag{Strategies}
A (total deterministic) {\em strategy} for player $i\in\Ag$ is a function of type $W^+\rightarrow\Act_i$. A {\em strategy profile for $\Gamma\subseteq\Ag$} is a tuple of strategies $\bm{s}=\angles{\bm{s}_i:i\in \Gamma}$, one for every member $i$ of $\Gamma$. 
We denote the set of all strategy profiles for $\Gamma$ by $S_\Gamma$.
Strategies as above apply to permanent coalitions. 
\oomit{In $\ATL^*$, $\atlD{\Gamma}\varphi$ means that there exists an $\bm{s}\in S_\Gamma$ such that all the infinite plays in which the members of $\Gamma$ act as prescribed by $\bm{s}$ satisfy the temporal condition $\varphi$. $\ATL^*$ model checking includes finding such $\bm{s}$, which can also be viewed as programs for $\Gamma$ that satisfy $\varphi$ as the specification.
}

\parag{Alternating-time Temporal Logic ($\ATL$, \cite{AHK97,AHK02})}
This is now a group of logics with the common defining feature being the game-theoretic modalities $\atlD{\Gamma}$ where $\Gamma\subseteq\Ag$. $\ATL^*$ syntax includes {\em state} formulas $\varphi$ and {\em path} formulas $\psi$. Path formulas are $\PLTL$ ($\LTL$ with past) formulas with $\ATL^*$ state formulas allowed to appear along with the atomic propositions. We use $\ATL^*$ with linear past. A recent study on $\ATL$s with linear past, with comprehensive references on the overall topic can be found in \cite{DBLP:journals/tcs/BozzelliMS20}.
The syntax of $\ATL^*$ with past can be defined by the BNFs:
\[\begin{array}{l}
\varphi::=\bot\mid p\mid\varphi\Rightarrow\varphi\mid\atlD{\Gamma}\psi\qquad
\psi::=\varphi\mid\psi\Rightarrow\psi\mid\Next\psi\mid\until\psi\psi\mid\yesterday\psi\mid\since\psi\psi
\end{array}\]
As it becomes clear below, the past operators facilitate reference to the players' objectives as formulated wrt the beginning of a play.
$\ATL^*$ semantics on CGMs comes in several variants; we choose the one in which a CGM $M$ and a finite play $\bm{w}\in R^\fin_M(w_I)$ need to be specified to define $\models$ on state formulas, and an infinite play, and a position $k<\omega$ in it are needed for path formulas. Given a $\Gamma\subseteq\Ag$, a $\bm{w}\in R^\fin_M(w_I)$ and an $\bm{s}\in S_\Gamma$, the set of the infinite continuations of $\bm{w}$ in which the players from $\Gamma$ act according to $\bm{s}$ is:
\[\outcome(\bm{w},\bm{s})\defeq\{\bm{v}\in R^\infin_M(\bm{w}):(\forall k\geq|\bm{w}|)(\exists \bm{b}\in\Act_{\Ag\setminus\dom \bm{s}})
(\bm{v}^k=o(\bm{v}^{k-1},\bm{s}(\bm{v}^0\ldots\bm{v}^{k-1})\cup \bm{b}))\}.\]
The definining clauses for $\models$ are as follows: 
{\small
\[\begin{array}{lcl}
M,\bm{w}\models p &\mbox{iff}& V(\bm{w}^{|\bm{w}|-1},p)\mbox{ for }p\in AP\\
M,\bm{w}\models \atlD{\Gamma}\varphi &\mbox{iff}& \mbox{there exists an }\bm{s}\in S_\Gamma
\mbox{ s. t. }M,\bm{v},|\bm{w}|-1\models\varphi \mbox{ for all }\bm{v}\in\outcome(\bm{w},\bm{s})\\
M,\bm{w},k\models\Next\varphi&\mbox{iff}&M,\bm{w},k+1\models\varphi\\
M,\bm{w},k\models\yesterday\varphi&\mbox{iff}&k>0\mbox{ and }M,\bm{w},k-1\models\varphi\\
M,\bm{w},k\models\until\varphi\psi&\mbox{iff}&\mbox{for some }n<\omega,\ M,\bm{w},k+n\models\psi \mbox{ and }M,\bm{w},k+m\models\varphi\mbox{ for all }m<n.\\
M,\bm{w},k\models\since\varphi\psi&\mbox{iff}&\mbox{for some }n\leq k,\ M,\bm{w},k-n\models\psi \mbox{ and }M,\bm{w},k-m\models\varphi\mbox{ for all }m<n.
\end{array}\]
}

\noindent
The clauses about $\bot$, $\top$ and $\Rightarrow$ are as usual; $\vee$, $\wedge$ and $\Leftrightarrow$ and the derived $\PLTL$ operators $\Diamond$, $\Box$, $\necpast$, $\past$, etc. are defined as usual too. The beginning of time is indicated by $\I\defeq\neg\yesterday\top$. The dual $\neg\atlD{\Gamma}\neg$ of $\atlD{\Gamma}$ is written $\atlB{\Gamma}$. $\ATL$'s $\atlD{\emptyset}$ and $\atlB{\emptyset}$ are equivalent to $\CTL$'s $\forall$ and $\exists$, respectively,  the genetic link between $\ATL$ and $\CTL$. In the sequel, for the sake of brevity, we write $\forall$ and $\exists$ for $\atlD{\emptyset}$ and $\atlB{\emptyset}$, respectively. For $\Gamma\subseteq\Ag$, $-\Gamma\defeq\Ag\setminus\Gamma$. 

For future state formulas $\varphi$, we put $\sem{\varphi}_M\defeq\{\bm{v}^{|\bm{v}|-1}:M,\bm{v}\models\varphi\}$. For future state $\varphi$, $M,\bm{v}\models\varphi$ does not impose conditions on
$\bm{v}^0\ldots\bm{v}^{|\bm{v}|-2}$. Therefore $M,\bm{v}\models\varphi$ is equivalent to $\bm{v}^{|\bm{v}|-1}\in \sem{\varphi}_M$ for such formulas.

We write $\Var(\varphi)$ and $\Subf(\varphi)$ for the set of the atomic propositions which occur in formula $\varphi$ and the set of $\varphi$'s subformulas, including $\varphi$ itself, respectively.

\oomit{
$\ATL$ is the subset of $\ATL^*$ in which the syntax of path formulas is restricted to 
\[
\varphi::=\Next\psi\mid\until\psi\psi\mid\Box\psi
\]
}

\parag{Unwinding CGMs}
Given a CGM $M$ as above, the {\em unwinding} $M^T\defeq\angles{W^T,w^T_I,\angles{\Act_i:i\in\Ag},o^T,V^T}$ of $M$ is defined as follows:
\[\begin{array}{lcllcl}
W^T & \defeq & W(\Act_\Ag\ W)^* & o^T(w^0\bm{a}^1\cdots \bm{a}^nw^n, \bm{b})&\defeq&w^0\bm{a}^1\cdots\bm{a}^nw^n\bm{b}o(w^n,\bm{b})\\
w^T_I & \defeq & w_I & V^T(w^0\bm{a}^1\cdots \bm{a}^nw^n,p) & \defeq & V(w^n,p)
\end{array}\]
$M^T$ and $M$ are bisimilar and $R_M^\infin(w_I)$ and $R_{M^T}^\infin(w^T_I)$ are isomorphic. Importantly, $o^T$ is invertible.

\parag{$\mathbf{QCTL}^*$}
The use of $\exists$ and $\forall$ as abbreviations for $\ATL$'s $\atlB{\emptyset}$ and $\atlD{\emptyset}$, respectively, gives rise to a sublanguage of $\ATL^*$ (with past) which coincides with the language of $\CTL^*$, for the same vocabulary. The interpretation of $\exists$ and $\forall$ which is entailed by their use as abbreviations in CGMs $M$ is consistent with the same mode of referencing the past in the $\CTL^*$ definitions of their {\em structure} semantics wrt models $\angles{W,w_I,R,V}$ where $W$, $w_I$, and $V$ are as in the CGM $M$, and the transition relation $R$ is defined using the CGM's outcome function $o$ by the equivalence $R(w,v)\leftrightarrow(\exists\bm{a}\in\Act_\Ag)(v=o(w,\bm{a}))$. 

Along with its structure semantics, $\QCTL^*$ admits also a {\em tree} semantics, which can be spelled out using the same defining clauses, except that $M^T$, the unwinding of the given model $M$, appears on the LHS of $\models$. The tree semantics differs from the structure one only wrt propositional quantification.
In the structure semantics,
\[\begin{array}{lcl}
M,\bm{w}\models\exists p\varphi&\mbox{iff}&\mbox{there exists an }X\subseteq W\mbox{ s. t. }
M^X_p\models\varphi.
\end{array}\]
With the values of quantified variables being varied on the states of the original $\angles{W,w_I,R,V}$, the recurrence of states along paths entails corresponding repeats of the variables' values. Unwindings display no such dependency as repeated occurrences of the original model's states along paths are replaced by distinct states, and the condition $X\subseteq W$ becomes replaced by $X\subseteq W^T$. In this work we refer mostly to unwindings, which can be regarded as using the tree semantics of arbitrary CGMs too. Model-checking and validity in $\QCTL^*$ are decidable on the class of tree-based Kripke models \cite{DBLP:conf/ausai/French01,French06} such as what underly unwindings $M^T$.

\section{Strategy Profiles with Temporary Coalitioning and a Vocabulary for Temporary Coalitions}

In this section we introduce our propositional vocabulary for making statements about strategic ability with the co-existence of temporary coalitions taken in account. In doing so, we build on a technique which enables the use of propositional variables for the naming of sets of transitions in modal languages. In $\QCTL^*$, along with the naming of strategies, this technique enables the expression of quantification over strategies on CGMs. In our approach, dedicated collections of propositional variables specify both sets of decisions in the established way and the coalition structure in place upon every transition. We upgrade the underlying semantic notion of strategy profile to incorporate coalition structure first. Let $\partition{\Gamma}$ stand for the set of the exhaustive partitions of $\Gamma$ into disjoint nonempty subsets.
\begin{definition}\label{sptcDefinition}
Assuming a CGM $M$ as usual, and a {\em strategy profile with temporary coalitions} ($\SPTC$) is a function of type $W^+\rightarrow\Act_\Ag\times\partition{\Ag}$.
\end{definition}
Informally, given a finite play, a $\SPTC$ specifies both a decision on how to continue the play and a partitioning of $\Ag$ into coalitions, the competing collective authors of that decision. Now let us explain how our vocabulary for $\SPTC$ builds on the conventions for naming decisions in a propositional temporal language.

\parag{Naming decisions and strategies by propositions in $\QCTL^*$} works as follows. In an arbitrary CGM $M$ with its components named as above, propositions $p$ define the sets of states $\sem{p}\defeq\{w\in W: V(w,p)\}$. In the unwinding $M^T$ of a CGM $M$, the invertibility of $o^T$ entails that any subset $S$ of $W^T\times\Act_\Ag$ can be recovered from the corresponding $\{o(w,\bm{a}):\angles{w,\bm{a}}\in S\}$, which is a set of states, and therefore can be denoted by a proposition. Finite plays can be determined from their last states in $M^T$, which renders memorylessness of strategies vacuous. Hence any strategy can be determined unambiguously from the set of the target states of the transitions it generates. A proposition which is true in states that are reachable by transitions that are consistent with the strategy and false elsewhere can serve to name the strategy. Hence $\exists s(\delta_\Gamma(s)\wedge\forall\Next(s\Rightarrow\varphi))$, a $\QCTL^*$ formula, in which $\delta_\Gamma(s)$ presumedly restricts $s$ to range over those $S\subseteq W^T\times\Act_\Ag$ which fit the description of strategies for $\Gamma$, expresses that $\Gamma$ can enforce $\varphi$ by implementing $s$ for one step. 

This observation combines beneficially with the facts that $\QCTL^*$ is decidable on tree-based Kripke models (aka execution trees) which are the unwindings of finite Kripke models \cite{DBLP:conf/ausai/French01,French06}, and, unsurprisingly, finite CGMs reduce to finite Kripke models, if the transition relation is defined by putting $R(w,v)\leftrightarrow(\exists\bm{a}\in\Act_\Ag)(v=o(w,\bm{a}))$, as mentioned in the Preliminaries section.

\parag{Specifying evolving coalition structure}
We specify evolving coalition structure by enhancing the above use of propositions. The propositional variables which we introduce serve to both name decisions and specify the coalition structure in place upon carrying out the decisions. Specifying a $\SPTC$ $s$ takes a collection $\bm{s}\defeq\angles{\bm{s}_\Gamma:\Gamma\subseteq\Ag}$ of variables. The extension $(M^T)_{\angles{\bm{s}_\Gamma:\Gamma\subseteq\Ag}}^{\angles{X_\Gamma(s):\Gamma\subseteq\Ag}}$ of $M^T$ in which the variables from $\bm{s}$ specify $s$ is defined by putting
\begin{equation}\label{xGamma}
X_\Gamma(s)\defeq\{o^T(\bm{w}^{|\bm{w}|-1},\bm{a}|_\Gamma\cup\bm{b}):\bm{w}\in R^\fin_{M^T}(w_I^T),\angles{\bm{a},C}= s(\bm{w}),\Gamma\in C,\bm{b}\in\Act_{-\Gamma}\}.
\end{equation}

This means that $(M^T)_{\angles{\bm{s}_\Gamma:\Gamma\subseteq\Ag}}^{\angles{X_\Gamma(s):\Gamma\subseteq\Ag}},\bm{v}\models\bm{s}_\Gamma$ indicates that $\bm{v}^{|\bm{v}|-1}$ can be reached from $\bm{v}^{|\bm{v}|-2}$ by $\Gamma$'s part of the $c$-decision for $\bm{v}^0\cdots\bm{v}^{|\bm{v}|-2}$, with the express condition that $\Gamma$'s part of this decision is a decision of $\Gamma$ {\em as a coalition}, and not a coincidental collection of unrelated decisions by the players from $\Gamma$.  

By this convention we can write, e.g., $\exists z(\delta_\Gamma(z)\wedge\forall\Next(z\wedge\bm{s}_{-\Gamma}\Rightarrow\varphi))$ to express that the members of $\Gamma$ can enforce $\varphi$ provided that the rest of the players make a $\bm{s}_{-\Gamma}$-move {\em as a coalition}. This condition on $-\Gamma$'s concurrent move may be of some consequence, if, e.g., the $\SPTC$ denoted by $\bm{s}$ meets the requirements a solution concept may be imposing, including requirements on the viability of coalitions. Then the $-\Gamma$ decisions which lead to $\bm{s}_{-\Gamma}$-states would range over those single step coalition agendas which, according to the adopted solution concept, appear to be sufficiently attractive to unite $-\Gamma$. 

In the rest of the paper $\bm{X}(s)$ abbreviates $\angles{X_\Gamma(s):\Gamma\subseteq\Ag}$; hence $(M^T)_{\bm{s}}^{\bm{X}(s)}$ stands for $(M^T)_{\angles{\bm{s}_\Gamma:\Gamma\subseteq\Ag}}^{\angles{X_\Gamma(s):\Gamma\subseteq\Ag}}$. As it becomes clear below, expressing comparisons between alternative $\SPTC$ may take $\CTL^*$ formulas which refer to more than one such system of variables of the form $\angles{\bm{s}_\Gamma:\Gamma\subseteq\Ag}$.

\parag{The apparent lack of semantics for coalition structure}
Note that, unlike decisions and strategies, coalition structure cannot be observed by examining the (global) transitions themselves in CGMs.\oomit{This is in agreement with the fact that a decision may be wise if chosen in coalition, and be reckless, however lucky, if commenced as the coincidental combination of individual decisions by the same players.} Therefore it is crucial to notice that the intended meaning of these dedicated variables is achieved only in extensions of $M^T$ by dedicated valuations for these variables, or by their {\em bound occurrences}, with the quantification also indicating that the coalitioning in question is {\em hypothetical}. This convention leads to averting the perceived necessity to upgrade CGMs for registering coalition structure by anything more than the valuation for the dedicated propositional variables introduced above.

\parag{Constraining the vocabulary to express well-formed coalition structure}
For a system $\bm{s}\defeq\angles{\bm{s}_\Gamma:\Gamma\subseteq\Ag}$ of variables to really specify how $\Ag$ is partitioned into disjoint coalitions, $\bm{s}$ must assign every player to a unique coalition, with the decisions of unallied players $i$ being denoted by the respective $\bm{s}_{\{i\}}$. Along with expressing that the reference state is the target of a transition in which all the decisions were $\bm{s}$-decisions, the formulas $\widetilde{\bm{s}}_\Gamma$ below express that each of the players from $\Gamma$ belongs to a unique coalition within $\Gamma$ and, furthermore, no coalitions include players from both inside and outside $\Gamma$. The following equivalence defines $\widetilde{\bm{s}}_\Gamma$, $\Gamma\subseteq\Ag$, recursively:
\begin{equation}\label{crat}
\widetilde{\bm{s}}_\Gamma
\Leftrightarrow
\bigwedge\limits_{\Gamma'\subseteq\Ag,\Gamma'\cap\Gamma\not=\emptyset,\Gamma'\setminus\Gamma\not=\emptyset}\neg \bm{s}_{\Gamma'}
\wedge 
\bigg(\bm{s}_\Gamma\wedge\bigwedge\limits_{\emptyset\subsetneq\Delta\subsetneq\Gamma}\neg \bm{s}_\Delta\vee\bigvee\limits_{\emptyset\subsetneq\Delta\subsetneq\Gamma}(\widetilde{\bm{s}}_\Delta\wedge\widetilde{\bm{s}}_{\Gamma\setminus\Delta})
\bigg).
\end{equation}
Since $-\Ag=\emptyset$, $\widetilde{\bm{s}}_\Ag$ only states that $\bm{s}$ is assigning every player to a unique coalition.\footnote{We prefer the above recursive definition of $\widetilde{\bm{s}}_\Gamma$ to a 'flat' encoding of the conditions on $\bm{s}$ as the recursive definition gives easy access to a useful corollary of the property, namely that either the whole of $\Gamma$ is a coalition, or some proper subset of its, e.g. some subset of a $\Delta$ from the RHS disjunct in the parentheses, is a coalition. Note that these conditions are not perfectly 'modular', as $\widetilde{\bm{s}}_\Gamma$ rules out coalitions which only overlap with $\Gamma$ and can extend among $\Ag\setminus\Gamma$ as well.} Next we spell out some more restrictions on $\bm{s}$ which further entail that a state which satisfies $\widetilde{\bm{s}}_\Ag$ ought to be the target state of a unique global decision, with this coalition structure in place. Namely, we restrict $\bm{s}_\Gamma$ to evaluate to the set of the target states of transitions with some specific $\bm{a}\in\Act_\Gamma$ as $\Gamma$'s projection of the involved global decision, in case $\Gamma$ appears in the current coalition structure. This is achieved by including in $AP$ names for the sets $\{o^T(w,\bm{a}):\bm{a}_i=a\}$ all $a\in\Act_i$, $i\in\Ag$. To enable this, assume that $\Act_i$, $i\in\Ag$, are pairwise disjoint subsets of $AP$, and $V^T$ is defined on them by the equivalences:
\[\begin{array}{l}
V^T(w^T_I,a)\leftrightarrow\bot,\qquad
V^T(w^0\bm{a}^1\cdots \bm{a}^{n+1}w^{n+1}, a)\leftrightarrow a=\bm{a}_i^{n+1}\mbox{ for all }a\in\Act_i,i\in\Ag.
\end{array}\]
In arbitrary concurrent game structures, the target sets of distinct actions by the same player need not be disjoint like in unwindings $M^T$ and recovering previous actions from their target states like above may be impossible. However, given an arbitrary finite CGM $M$, an expansion $\bar{M}$ of $M$ can be defined, which is still finite, has an outcome function which is invertible wrt decisions, and has unwinding $\bar{M}^T$ which is isomorphic to $M^T$, the unwinding of the original $M$. $\bar{M}\defeq\angles{\bar{W},\bar{w}_I,\angles{\Act_i:i\in\Ag},\bar{o},\bar{V}}$ can be viewed as extending $M$ to store (only) the latest global decisions in their target states, as opposed to $M^T$, which stores whole histories:
\[\begin{array}{lcllcl}
\bar{W}&\defeq&W\times(\Act_\Ag\cup\{*\}) &
\bar{V}(\angles{w,\bm{a}},p)&\defeq& V(w,p)\mbox{ for }p\in AP\setminus\bigcup\limits_{i\in\Ag}\Act_i \\
\bar{w}_I&\defeq&\angles{w_I,*} &
\bar{V}(\angles{w,\bm{a}},a)&\defeq&a=\bm{a}_i\mbox{ for }a\in\Act_i,i\in\Ag\\
\bar{o}(\angles{w,\bm{b}},\bm{a})&\defeq&\angles{o(w,\bm{a}),\bm{a}}
\end{array}
\]
That is why $M^T$, with the valuations of the action-naming atomic propositions as defined, for finite $M$, is still the unwinding of a finite CGM. We highlight this fact because it is relevant for the decidability of satisfaction of formulas in it. 

Given an $\bm{a}\in\Act_\Gamma$ for some $\Gamma\subseteq\Ag$, let $\hat{\bm{a}}\defeq\bigwedge\limits_{i\in\Gamma}\bm{a}_i$. Let
\[\tilde{\delta}(\bm{s})\defeq
\forall\Box\exists\Next\widetilde{\bm{s}}_\Ag\wedge
\bigwedge\limits_{\Gamma\subseteq\Ag}
\forall\Box\bigg(
\forall\Next\neg\bm{s}_\Gamma\vee
\bigvee\limits_{\bm{a}\in\Act_\Gamma}\forall\Next(\bm{s}_\Gamma\Leftrightarrow\hat{\bm{a}})
\bigg)
\]
For any collection $\bm{Y}\defeq\angles{Y_\Gamma:\Gamma\subseteq\Ag}$ of subsets of $W^T\setminus\{w_I^T\}$, it can be shown that $(M^T)_{\bm{s}}^{\bm{Y}},w_I^T\models\tilde{\delta}(\bm{s})$
iff there exists a $\SPTC$ $s$ for $M^T$ such that $\bm{Y}=\bm{X}(s)$ where $\bm{X}(s)$ is as in (\ref{xGamma}).

\parag{Some basic expressions in $\SPTC$ vocabularies}
Given a $\SPTC$ $s$, $(M^T)_{\bm{s}}^{\bm{X}(s)},\bm{w},0\models\Box\Next\widetilde{\bm{s}}_\Ag$ means that play $\bm{w}$ in $M^T$ is consistent with $s$. $\SPTC$ $s$ enforces $\varphi$, if
\begin{equation}\label{sEnforcesVarphi}
(M^T)_{\bm{s}}^{\bm{X}(s)},w_I^T\models\forall\bigg(\Box\Next\widetilde{\bm{s}}_\Ag\Rightarrow\varphi\bigg).
\end{equation}
$\SPTC$ $s$ enables $i\in\Ag$ to achieve $\varphi$, if
\begin{equation}\label{iFollowsStrategy}
(M^T)_{\bm{s}}^{\bm{X}(s)},w_I^T\models\forall\bigg(\Box\Next\bigvee\limits_{i\in\Gamma\subseteq\Ag}\bm{s}_\Gamma\Rightarrow\varphi\bigg).
\end{equation}
In words, if every step of a play is consistent with the agenda of the $s$-coalition where $i$ belongs, then that play is bound to satisfy $\varphi$, regardless of the agendas of the co-existing temporary coalitions. In plays which satisfy 
$\Box\Next\widetilde{\bm{s}}_\Ag$, $i$ belongs to exactly one coalition at any time. Importantly, it is not guaranteed that the hypothetical allies of $i$ would be really interested in participating in the designated coalitions.

\parag{Strategy contexts with temporary coalitions}
Given our conventions about the use of systems $\bm{s}=\angles{\bm{s}_\Gamma:\Gamma\subseteq\Ag}$ of propositional variables for naming $\SPTC$, we can define strategic ability of (permanent) coalition $\Gamma$ in the {\em context} of a $\Delta\subseteq\Ag\setminus\Gamma$, implementing (their part of) a $\SPTC$ denoted by $\bm{s}$ by putting
\begin{equation}\label{defPath}
\atlD{\Gamma}^{\bm{s}}_\Delta\varphi\defeq
\atlD{\Gamma}(\Box\Next\widetilde{\bm{s}}_{\Delta}\Rightarrow\varphi).
\end{equation}
As expected, we write the dual as $\atlB{\Gamma}^{\bm{s}}\varphi\defeq\neg\atlD{\Gamma}^{\bm{s}}\neg\varphi$. This can be regarded as a temporary coalition upgrade of the semantics of {\em strategy contexts} from \cite{WaltherHW07,BrihayeLLM09,WangHY11}.

\parag{Related Work on the Specification of Strategies by Propositional Variables}  
The technique we use for the naming of decisions and strategies is a mixture of folklore an authored work. For instance, the transformation of CGMs involved was key to establishing the equivalence between the \cite{AHK97} alternating transition systems based semantics of $\ATL$s and the \cite{AHK02} CGM-based one in \cite{GJ04}. In automata theory, a similar transformation is known as the passage from automata with labels on the transitions to automata with (those same) labels appearing on the respective transitions' target states; the latter are called {\em state-based} automata. A different technique of using propositions and quantification for encoding strategies and expressing their existence was proposed in \cite{DBLP:conf/atva/Pinchinat07}.
In \cite{DBLP:journals/corr/abs-1303-0794}, we used a similar naming technique to show how validity in (a subset of epistemic) $\ATL$ reduces to validity in (epistemic) $\CTL$. In \cite{DBLP:conf/clima/GuelevD12} we applied the technique to epistemic $\ATL$ with strategy contexts, to essentially identify the contribution of individual coalition members and the contribution of the players from the strategy context towards the considered objectives. Model checking complete information $\ATL_{\mathit{sc}}$ was shown to be reducible to validity in $\QCTL^*$ by introducing atomic propositions to name states and decisions and encoding the CGM's outcome function as a formula in terms of these propositions in \cite{DBLP:conf/concur/LopesLM12}.

\section{Preference in Concurrent Game Models}
\label{preferenceSection}

We assume the preference relations $<_i$ of individual players $i$ to be strict partial orders on infinite plays, with unrelated plays being of the same value to the respective players. In much of the literature preference is a pre-order; a comprehensive discussion on modeling preference can be found in \cite{HanssonPreferenceLogic}.
\oomit{If preference is modeled as a {\em pre}order $\lesssim$, plays $\bm{v}$ and $\bm{w}$ can be incomparable in two ways: in case both $\bm{v}\lesssim\bm{w}$ and $\bm{w}\lesssim\bm{v}$, and in case neither $\bm{v}\lesssim\bm{w}$, nor $\bm{w}\lesssim\bm{v}$.} 

To facilitate algorithmic methods, we require the relations
\[\bm{v}\sim_i\bm{w}\defeq(\forall\bm{u}\in R^\infin_M(w_I))((\bm{u}<_i\bm{v}\leftrightarrow \bm{u}<_i\bm{w})\wedge (\bm{v}<_i\bm{u}\leftrightarrow \bm{v}<_i\bm{w}))\]
of indiscernibility wrt $<_i$, $i\in\Ag$, to partition $R^\infin_M(w_I)$ into finitely many $\LTL$-definable classes. 

We denote the sets of the formulas which define these indiscernibility classes by $\Theta_{I,i}$, $i\in\Ag$, where $I$ stands for {\em I}nitial like in $w_I$. We require distinct $\theta',\theta''\in\Theta_{I,i}$ to define disjoint sets of plays: $M,w_I\models\forall\neg(\theta'\wedge\theta'')$. To incorporate the preference relations $<_i$, we consider {\em extended} CGMs $M$ of the form
\begin{equation}\label{extCGM1}
\angles{W,w_I,\angles{\Act_i:i\in\Ag},o,\angles{<_i:i\in\Ag},V}
\end{equation}
Given $<_i$, we assume that the sets $\Theta_{I,i}$, become available automatically. Then tuples of the form 
\begin{equation}\label{extCGM2}
\angles{W,w_I,\angles{\Act_i:i\in\Ag},o,\angles{\angles{\Theta_{I,i},<_i}:i\in\Ag},V}
\end{equation}
with the partial orders $<_i$ now defined on the classes of plays specified by $\theta\in\Theta_{I,i}$ can serve as extended CGMs too. In the sequel we tacitly assume an arbitrary fixed extended CGM with the preferences of the players in it represented either as in (\ref{extCGM1}) or as in (\ref{extCGM2}).

The restriction to $\LTL$-definable properties is meant to provide a match with the expressive power of the TL language. Requiring the classes to be just regular $\omega$-languages is no less reasonable. This more general setting has been investigated in \cite{DBLP:conf/fossacs/BouyerBMU12}. That work also presents a compelling variety of ways to combine primitive objectives into compound ones.
Temporary coalitions are natural to occur given the ordering of the objectives in the example below.

\begin{example}\label{example1}
$1$ is a sworn enemy of $2$ but feels for its own life even more strongly. Then, if $p_i$ stands for {\em $i$ perishes}, $1$'s preferences can be expressed as follows:
\begin{equation}\label{exampleOrdering1}
(\Diamond p_1\wedge\Box\neg p_2)
<_1(\Diamond p_1\wedge\Diamond p_2)
<_1(\Box\neg p_1\wedge\Box\neg p_2)
<_1(\Box\neg p_1\wedge\Diamond p_2).
\end{equation}
With $3$ sworn enemies we have:
\begin{equation}\label{exampleOrdering2}
{\small\begin{array}{ccccc}
& & \left(\begin{array}{l}\Diamond p_1\wedge\\ \Diamond p_2 \wedge\\ \Box\neg p_3\end{array}\right)
<_1\left(\begin{array}{l}\Box\neg p_1\wedge\\ \Box\neg p_2 \wedge \\ \Box\neg p_3\end{array}\right)
<_1\left(\begin{array}{l}\Box\neg p_1\wedge\\ \Diamond p_2 \wedge\\ \Box\neg p_3\end{array}\right)\\
\left(\begin{array}{l}\Diamond p_1\wedge\\ \Box\neg p_2 \wedge\\ \Box\neg p_3\end{array}\right)
&
<_1&\mbox{\rotatebox{90}{$<_1$}} &<_1& \left(\begin{array}{l}\Box\neg p_1\wedge\\ \Diamond p_2 \wedge\\ \Diamond p_3\end{array}\right)\\
& &\left(\begin{array}{l}\Diamond p_1\wedge \\ \Box\neg p_2 \wedge\\ \Diamond p_3\end{array}\right)
<_1\left(\begin{array}{l}\Diamond p_1\wedge\\ \Diamond p_2 \wedge\\ \Diamond p_3\end{array}\right)
<_1\left(\begin{array}{l}\Box\neg p_1\wedge\\ \Box\neg p_2 \wedge\\ \Diamond p_3\end{array}\right)
\end{array}}
\end{equation}
Now, e.g., $1$ and $2$ may conspire to eliminate $3$, but their relations are bound to sour afterwards. 
\end{example}

For non-singleton coalitions $\Gamma$, we write
\oomit{$\bm{w}<_\Gamma\bm{v}\defeq\bigwedge\limits_{i\in\Gamma}\bm{w}<_i\bm{v}$,}
$\Theta_{I,\Gamma}\defeq\bigg\{\bigwedge\limits_{i\in\Gamma}\theta_i:\theta_i\in\Theta_{I,i},i\in\Gamma\bigg\}.$ We refer to the formulas from $\Theta_{I,\Gamma}$ as {\em $\Gamma$'s objectives}.

In the temporal language, we introduce the binary operators $<_i$ and $\doubleNotLess_i$, $i\in\Ag$, as follows:
\[\begin{array}{lll}
M,\bm{v}\models\varphi_1 <_i\varphi_2 & \mbox{iff} & \bm{w}_1<_i\bm{w}_2\mbox{ for all }\bm{w}_k\in R^\infin_M(\bm{v})\mbox{ s. t. }
M,\bm{w}_k, |\bm{v}|\models\varphi_k,\ k=1,2.\\
M,\bm{v}\models\varphi_1 \doubleNotLess_i\varphi_2 & \mbox{iff} & \bm{w}_1\not<_i\bm{w}_2\mbox{ for all }\bm{w}_k\in R^\infin_M(\bm{v})\mbox{ s. t. }
M,\bm{w}_k, |\bm{v}|\models\varphi_k,\ k=1,2.
\end{array}\]
For the sake of simplicity, we allow only $\PLTL$ formulas to be operands of $<_i$ and $\doubleNotLess_i$, with no embedded branching time constructs. 

In words, $M,\bm{v}\models\varphi_1 <_i\varphi_2$ means that $i$ prefers $\varphi_2$ plays to $\varphi_1$ ones. Unlike $\neg(\varphi_1 <_i\varphi_2)$,  $\varphi_1 \doubleNotLess_i\varphi_2$ holds only if no $\varphi_1$ play is preferable to a $\varphi_2$ one. We abbreviate $\bigwedge\limits_{i\in\Gamma}\varphi<_i\psi$ and  $\bigwedge\limits_{i\in\Gamma}\varphi\doubleNotLess_i\psi$ by $\varphi<_\Gamma\psi$ and $\varphi\doubleNotLess_\Gamma\psi$, respectively.

According to \cite{van2005preference}, the non-temporal archetype of $<_i$ can be traced back to \cite{vonWrightPreference}. It is just one of 8 preference operators featured in \cite{van2005preference}, the one about the case in which {\em all} the $\varphi_1$-plays and {\em all} the $\varphi_2$-plays are related by the {\em strict} ordering $<_i$ on individual plays. The remaining 7 operators are defined by changing {\em all} to {\em some}, and changing $<_i$ to the non-strict $\lesssim_i$. Similar operators have been investigated in the setting of discrete contact spaces and relational syllogistic in \cite{DBLP:journals/jancl/BalbianiTV07,DBLP:journals/jolli/IvanovV12}. 
In \cite{Lorini2010}, the state space consists of the global strategy profiles. The graded preference relations $\gtrsim_i^k$, $k\leq n$, relate the reference profile to the profiles which are of {\em quality $k$ or higher} to player $i$, $k\in\{0,\ldots,n\}$, and are used to define unary modalities in the standard way. The binary operator $\lesssim_i$ is defined by the clause
\[\varphi\lesssim_i\psi\defeq\bigwedge\limits_{k\leq n}\angles{\gtrsim_i^k}\varphi\Rightarrow\angles{\gtrsim_i^k}\psi.\]

\section{Encoding Solution Concepts in $\ATL^*$}
\label{encodingSolutionConcepts}

To conveniently refer to objectives as they stand at the beginning of time $\I\defeq\neg\yesterday\top$, we write 
\[[\theta]\defeq\past(\I\Rightarrow\theta)\ .\]
In this section we use the notation for temporary coalitions to express some solution concepts in $\QCTL^*$ with the preference operator $<_i$.

\parag{The Core}
According to \cite{CooperativeGameTheory2012}, the {\em core} of a game is a set of decisions which is preferable to all coalitions. Let $\theta_*$ define the plays which can be the outcome of carrying out a core $\SPTC$. Then, assuming that $\theta_*$ is achieved by $\SPTC$ $r$, we must have, for any alternative $\SPTC$ $s$, 
\[(M^T)_{\bm{r},\ \ \bm{s}}^{\bm{X}(r),\bm{X}(s)},w_I^T\models
\forall(\Box\Next\widetilde{\bm{r}}_{\Ag}\Rightarrow\theta_*)
\wedge\bigwedge\limits_{\theta_\circ\in\Theta_{I,\Ag}}\forall(\Box\Next\widetilde{\bm{s}}_{\Ag}\Rightarrow\theta_\circ)\Rightarrow\bigwedge\limits_{i\in\Ag}
[\theta_\circ\wedge\neg\theta_*]<_i[\theta_*]\]

\parag{Domination with Temporary Coalitions}
With no assumptions on coalitioning, a strategy $s\in S_{\{i\}}$ is {\em dominant}, if it brings achievements which, in the view of player $i$, are better than those guaranteed by any other strategy, regardless of what the other players do. One simple way to generalize this to temporary coalitions is to require the outcome to be better for the coalition members, regardless of what the non-members of the coalitions entered by $i$ in the various parts of plays do. To express that a certain $\SPTC$ $r$ is dominant for $i$, we need to state that any strategy profile $s$ is either the same as $r$ to $i$, (assuming the resilience of the coalitions in which $i$ participates, or performs worse than $r$ for $i$: 
\[(M^T)_{\bm{r},\ \ \bm{s}}^{\bm{X}(r),\bm{X}(s)},w_I^T\models
\left(\begin{array}{l}
\forall\Box\bigwedge\limits_{i\in\Gamma\subseteq\Ag}(\bm{s}_\Gamma\Leftrightarrow \bm{r}_\Gamma)\vee
\\

\bigwedge\limits_{\theta'\in\Theta_{I,i}}\left(\begin{array}{l}
\forall\bigg(\Box\Next\bigvee\limits_{i\in\Gamma\subseteq\Ag}\bm{s}_\Gamma\Rightarrow[\theta']\bigg)
\Rightarrow\\
\bigvee\limits_{\theta''\in\Theta_{I,i}}[\theta']<_i[\theta'']\wedge
\forall\bigg(\Box\Next\bigvee\limits_{i\in\Gamma\subseteq\Ag}\bm{r}_\Gamma\Rightarrow[\theta'']\bigg)\end{array}
\right)
\end{array}\right)
\]
Note that, since the variables from $\bm{s}$ denote a $\SPTC$, the variables $\bm{s}_\Gamma$ from $\Box\Next\bigvee\limits_{i\in\Gamma\subseteq\Ag}\bm{s}_\Gamma$ can each be true for at most one $\Gamma$ at any state. The same holds about $\bm{r}_\Gamma$. Importantly, the above formula does not state that the players who are joining $i$ for the various steps of plays according to $r$ are interested in doing so.

For an $r$ to be dominant for all $i\in\Ag$, an arbitrary profile $s$ should either coincide with $r$, or perform worse than $r$ for all $i\in\Ag$, i.e., in formulas, $(M^T)_{\bm{r},\ \ \bm{s}}^{\bm{X}(r),\bm{X}(s)},w_I^T\models D(\bm{s},\bm{r})$ where 
\[
D(\bm{s},\bm{r})\defeq \left(\begin{array}{ll}\forall\Box\bigwedge\limits_{\Gamma\subseteq\Ag}\bm{s}_\Gamma\Leftrightarrow\bm{r}_\Gamma\vee\\
\bigwedge\limits_{\bm{\theta}'\in\prod\limits_{i\in\Ag}\Theta_{I,i}}
\left(\begin{array}{l}
\bigwedge\limits_{i\in\Ag}\forall\bigg(\Box\Next\bigvee\limits_{i\in\Gamma\subseteq\Ag}\bm{s}_\Gamma\Rightarrow[\bm{\theta}'_i]\bigg)\Rightarrow\\
\bigvee\limits_{\bm{\theta}''\in\prod\limits_{i\in\Ag}\Theta_{I,i}}\bigg(
\bigwedge\limits_{i\in\Ag}[\bm{\theta}_i']<_i[\bm{\theta}_i'']
\wedge
\bigwedge\limits_{i\in\Ag}\forall\bigg(\Box\Next\bigvee\limits_{i\in\Gamma\subseteq\Ag}\bm{r}_\Gamma\Rightarrow[\bm{\theta}''_i]\bigg)\bigg)
\end{array}\right)
\end{array}\right)\]
Let $\{\Gamma_0,\ldots,\Gamma_{2^{|\Ag|-1}}\}\defeq{\mathcal P}(\Ag)$. Then
the existence of a dominant profile can be expressed in $\QCTL^*$ with $<_i$ by the condition 
\[M^T,w_I^T\models\exists\bm{r}_{\Gamma_0}\ldots\exists\bm{r}_{\Gamma_{2^{|\Ag|-1}}}
(
\tilde{\delta}(\bm{r})\wedge
\forall\bm{s}_{\Gamma_0}\ldots\forall\bm{s}_{\Gamma_{2^{|\Ag|-1}}}(\tilde{\delta}(\bm{s})\Rightarrow D(\bm{s},\bm{r}))
)
.\]


\section{Axioms for $<_i$ and $\doubleNotLess_i$}
\label{axioms}

The axioms in this section are valid without the assumption on $<_i$ to be partitioning $R^\infin_M(w_I)$ into finitely many $\LTL$-definable indiscernibility classes. 
$P1$ expresses extensionality and is the form of ${\bf K}$ that applies to binary modalities. The axioms $P2$ state that $<_\Gamma$ and $\doubleNotLess_\Gamma$ are closed under disjunctions on both sides. $P3$ and $P4$ state that $<_i$ is irreflexive and transitive, respectively. The axioms $P5$ state that $<_\Gamma$ and $\doubleNotLess_\Gamma$ trivially hold, if one of the operands is $\bot$. \ref{P6} states that players stick to their preferences. Below $\sigma$ stands for $<$ or $\doubleNotLess$.
\begin{align}
\label{P1}\tag{P1}
& \varphi_1\sigma_\Gamma\psi_1\wedge\forall\Next(\varphi_2\Rightarrow\varphi_1)
\wedge\forall\Next(\psi_2\Rightarrow\psi_1)\Rightarrow\varphi_2\sigma_\Gamma\psi_2\\
\label{P2}\tag{P2}
& \varphi_1\sigma_\Gamma\psi\wedge\varphi_2\sigma_\Gamma\psi\Leftrightarrow(\varphi_1\vee\varphi_2)\sigma_\Gamma\psi,\qquad\varphi\sigma_\Gamma\psi_1\wedge\varphi\sigma_\Gamma\psi_2\Leftrightarrow\varphi\sigma_\Gamma(\psi_1\vee\psi_2)\\
\label{P3}\tag{P3}
& \varphi<_\Gamma\psi\Rightarrow\forall\Next\neg(\varphi\wedge\psi)\qquad\varphi<_\Gamma\psi\Rightarrow\psi\doubleNotLess_\Gamma\varphi\qquad
\varphi\doubleNotLess_\Gamma\psi\Rightarrow\neg(\varphi<_\Gamma\psi)\\
\label{P4}\tag{P4}
& \varphi<_\Gamma\psi\wedge\psi<_\Gamma\chi\wedge\exists\Next\psi\Rightarrow\varphi<_\Gamma\chi\\
\label{P5}\tag{P5}
& \bot\sigma_\Gamma\varphi,\qquad \varphi\sigma_\Gamma\bot\\
\label{P6}\tag{P6}
& [\varphi]\sigma_\Gamma[\psi]\Leftrightarrow\forall\Box([\varphi]\sigma_\Gamma[\psi])
\end{align}
According to \cite{van2005preference}, $P2$ and a variant of $P4$ without the conjunctive member $\exists\Next\psi$ can be traced back to \cite{vonWrightPreference}. Interestingly, that variant of $P4$ is unsound in our semantics: From $\varphi_1<_\Gamma\bot$ and $\bot<_\Gamma\varphi_2$, $P4$ with the $\exists\Next\bot$ (which is false) deleted entails $\varphi_1<_\Gamma\varphi_2$.

The meaning of $[\theta]$ after finite plays which render $\theta$ either ultimately failed ($M^T,\bm{w}\models\forall\neg[\theta]$) or achieved ($M^T,\bm{w}\models\forall[\theta]$) reflects removing from sight objectives that are no longer relevant. The example below illustrates this.  
\begin{example}
Consider the 3-player setting from Example \ref{example1} with the ordering of $1$'s objectives given in (\ref{exampleOrdering2}). 
Suppose that $p_3$ occurs.  Then the objectives from (\ref{exampleOrdering2}) which include $\Box\neg p_3$ become forfeited. The others simplify to objectives which no longer mention $p_3$ and are ordered as in (\ref{exampleOrdering1}). Formally, \ref{P1} entails that, e.g., $[\Diamond p_1\wedge\Box\neg p_2\wedge\Diamond p_3]<_1
[\Diamond p_1\wedge\Diamond p_2\wedge\Diamond p_3]$ is equivalent to
\[((\past p_3\vee\necpast\neg p_3\wedge\Diamond p_3)\wedge[\Diamond p_1\wedge\Box\neg p_2])
<_i
((\past p_3\vee\necpast\neg p_3\wedge\Diamond p_3)\wedge[\Diamond p_1\wedge\Diamond p_2]).\]
Now a somewhat longer deduction entails the validity of
\[p_3\Rightarrow ([\Diamond p_1\wedge\Box\neg p_2\wedge\Diamond p_3]<_1
[\Diamond p_1\wedge\Diamond p_2\wedge\Diamond p_3]\Leftrightarrow [\Diamond p_1\wedge\Box\neg p_2]<_1
[\Diamond p_1\wedge\Diamond p_2]).\]
Similarly, once $p_2$ occurs too, (\ref{exampleOrdering1}) simplifies to  
 just $[\Diamond p_1]<_1[\Box\neg p_1]$ because then the continuations of the play do not satisfy $[\Diamond p_1\wedge\Box\neg p_2]$ and $[\Box\neg p_1\wedge\Box\neg p_2]$, and $\Diamond p_2$ from $[\Diamond p_1\wedge\Diamond p_2]$ and $[\Box\neg p_1\wedge\Diamond p_2]$ simplifies to $\top$. If $\Box(\past p_2\wedge \past p_3\Rightarrow p_1\vee\Box\neg p_1)$, {\em nothing can kill $1$ but $2$ or $3$}, definitely a solid reason for $1$ to hate $2$ and $3$, is valid in the model, then $1$ can relax after $p_2$ and $p_3$ occur.
\end{example}

\section{Reasoning with Finitely Many Given Objectives}
\label{finitelyManyObjectives}

As stated in Section \ref{preferenceSection}, we assume $<_i$ to be partitioning plays into finitely many classes of pairwise indiscernible plays, and we assume these classes to be definable by some given finite set of formulas $\Theta_{I,i}$ for every $i\in\Ag$.
Consider the axioms
\begin{align}
\label{O1new}\tag{O1} &
\exists\Next(\varphi\wedge[\theta])\Rightarrow(\varphi\sigma_\Gamma\psi\Leftrightarrow(\varphi\vee[\theta])\sigma_\Gamma\psi),\qquad \exists\Next(\psi\wedge[\theta])\Rightarrow(\varphi\sigma_\Gamma\psi\Leftrightarrow\varphi\sigma_\Gamma(\psi\vee[\theta]))\\
\label{O3}\tag{O2}&
\bigvee\limits_{\theta\in\Theta_{I,\Gamma}}[\theta],\qquad
\bigwedge\limits_{\theta_1,\theta_2\in\Theta_{I,\Gamma}}
\forall\Box([\theta_1]\Leftrightarrow[\theta_2])\vee\forall\Box\neg([\theta_1]\wedge[\theta_2])\\
\label{O4}\tag{O3}&
[\theta_1]<_\Gamma[\theta_2],\mbox{ resp. }[\theta_1]\doubleNotLess_\Gamma[\theta_2],\mbox{ for }\theta_1,\theta_2\mbox{ such that }\theta_1<_\Gamma\theta_2,\mbox{ resp. }\theta_1\not<_\Gamma\theta_2,\mbox{ is given.}
\end{align}
for $\theta,\theta_1,\theta_2\in\Theta_{I,\Gamma}$, $\Gamma\subseteq\Ag$, $\sigma\in\{<,\doubleNotLess\}$.

The axioms from Section \ref{axioms} and this section are complete for $<_i$ and $\doubleNotLess_i$ in $\CTL^*$ with $<_i$ and $\doubleNotLess_i$, $i\in\Ag$, relative to validity in (just) $\CTL^*$. The full deductive power of the axioms about $<_i$ and $\doubleNotLess_i$ from Section \ref{axioms} is not really necessary for the case of finitely many given objectives. The completeness proof is based on the possibility to eliminate the occurrences of $<_i$, $\doubleNotLess_i$ and this way show that formulas with $<_i$ have $<_i$- and $\doubleNotLess_i$-free equivalents. Hence the small model property is inherited from $\CTL^*$ without $<_i$ and $\doubleNotLess_i$.

\begin{lemma}\label{usefulThms}
Let $\varphi_1$ and $\varphi_2$ be $\PLTL$ formulas, and let $\Gamma\subseteq\Ag$. Then the formulas below, $\sigma\in\{<,\doubleNotLess\}$, are derivable in $\CTL^*$ by the axioms from Section \ref{axioms} and this section:
{\small 
\begin{align}
\label{T5}\tag{$\mathrm{E}_\sigma$}
& \varphi_1\sigma_\Gamma\varphi_2\Leftrightarrow\bigvee\limits_{\Theta_1,\Theta_2\subseteq\Theta_{I,\Gamma}}\bigwedge\limits_{k=1,2}\bigg(\bigwedge\limits_{\theta_k\in\Theta_k}\exists\Next([\theta_k]\wedge\varphi_k)
\wedge
\bigwedge\limits_{\theta_k\in\Theta_{I,\Gamma}\setminus\Theta_k}\forall\neg\Next([\theta_k]\wedge\varphi_k)\bigg)\wedge\bigwedge\limits_{\theta_1\in\Theta_1,\theta_2\in\Theta_2}[\theta_1]\sigma_\Gamma[\theta_2]
\end{align}
}
\end{lemma}

The theorem below follows from the possibility to eliminate arbitrary uses of $<_i$ and $\doubleNotLess_i$ by means of \ref{T5} and resolve uses with the operands being designated objectives using \ref{O4}.
\begin{theorem}[relative completeness]
Every formula in $\CTL^*$ with $<_i$ and $\doubleNotLess_i$ has a $<_i$- and $\doubleNotLess_i$-free equivalent in $\CTL^*$. The equivalence is derivable by the axioms and rules from Section \ref{axioms} and this section.  
\end{theorem}
\oomit{
\begin{proof}
By virtue of \ref{T3}, subformulas with $<_\Gamma$ or $\doubleNotLess_\Gamma$ as the main connective can be replaced by formulas where $<_\Gamma$ and $\doubleNotLess_\Gamma$ are used only with operands from $\Theta_{I,\Gamma}$. However, by virtue of \ref{P6} and \ref{O4}, the latter are equivalent to $\top$ or $\bot$, depending only on the choice of the operands from $\Theta_{I,\Gamma}$.
\end{proof}
}

\oomit{\noindent
Note that, despite that the theorem applies to CGMs and $\ATL^*$, the equivalents it claims to exist are actually $\CTL^*$ formulas. Even with $\ATL^*$ formulas, if $\Act_i$, $i\in\Ag$ are known, then the equivalence (\ref{atlDelimination}) allows an equivalent to the given formula to be found in $\QCTL^*$.
}

\section{Reasoning with No Model-supplied System of Objectives}

This section is about the relative completeness of the axioms about $<_i$ from Section \ref{axioms}, with no assumptions on the properties of indiscernibility wrt $<_i$ in the considered models. The result shows that the satisfiability of any given a formula $\varphi$ in $\CTL^*$ in extended CGMs with the indiscernibility induced by $<_i$ not necessarily having finite index, is equivalent to the satisfiability of $\varphi$ in an extended CGM with a finite system of objectives which can be determined from $\varphi$. These objectives are boolean combinations of the operands of the occurrences of $<_i$ in $\varphi$.

\oomit{
Mind that the result does not include $\doubleNotLess_i$. Axiomatizing $<_i$ and $\doubleNotLess_i$ together appears to require a rule similar to Gabbay's irreflexivity rule \cite{Gabbay81}, to implicitly axiomatise the introduction of negative occurrences of the existential quantifier in the following characterisation of $\not<_i$: 
\[\models\chi'\not<_i\chi''\Leftrightarrow\exists P((\chi''\wedge P)<_i(\chi'\wedge P))\]
In the formula on the RHS of $\Leftrightarrow$ above $P$ is assumed to range over the continuations of the reference finite play. Since atomic propositions in $\LTL$ evaluate to sets of states, and not plays, $P$ cannot be treated as a propositional variable. This makes it difficult to port the rule-based technique from, e.g., \cite{DBLP:journals/jancl/BalbianiTV07,DBLP:journals/jolli/IvanovV12}.
}

We exclude the axioms from Section \ref{finitelyManyObjectives}. 
Instead we add two axioms about the interaction of $<_i$ and $\doubleNotLess_i$ with the {\em separated normal form} \cite{Gab89} of $\PLTL$ formulas and the 
{\em guarded normal form} in (future) $\LTL$. Let $\sigma\in\{<,\doubleNotLess\}$. Let $\{g_0,\ldots,g_{2^{|AP|}-1}\}$ be the set of all the conjunctions of the form $\bigwedge\limits_{p\in AP}\varepsilon_p p$ where $\varepsilon_p$ is either $\neg$ or nothing for each $p\in AP$. Then 
\begin{align}\label{P10}\tag{P7}
\bigg(\bigwedge\limits_{k'}\yesterday\pi_{k'}'\Rightarrow\varphi_{k'}'\bigg)\sigma_\Gamma\bigg(\bigwedge\limits_{k''}\yesterday\pi_{k''}''\Rightarrow\varphi_{k''}''\bigg)\Leftrightarrow\bigwedge\limits_{k',k''}(\pi_{k'}'\wedge\pi_{k''}''\Rightarrow\varphi_{k'}'\sigma_\Gamma\varphi_{k''}'')\\
\label{inductiveP}\tag{P8}
\bigg(\bigvee\limits_{k<2^{|AP|}}g_k\wedge\Next\varphi^t_k\bigg)\sigma_\Gamma\bigg(\bigvee\limits_{k<2^{|AP|}}g_k\wedge\Next\psi^t_k\bigg)\Leftrightarrow\forall\Next\bigg(\bigvee\limits_{k<2^{|AP|}} g_k\wedge\varphi^t_k\sigma_\Gamma\psi^t_k\bigg)
\end{align}
Formulas $\varphi$ which are consistent with \ref{P1},\ldots,\ref{inductiveP}, are satisfiable in a model where the preference relations partition the set of the infinite plays into finitely many classes of pairwise indiscernible plays. This means that our relative completeness result applies to the class of models of the form (\ref{extCGM2}) with finite systems of objectives too. Note that propositional quantification and $\doubleNotLess_i$ are not included in the completeness result below. Proving the exact form of the completeness result involves a finite set of instances of the axioms \ref{P1},\ldots,\ref{inductiveP}, which depends on the formula $\varphi$ whose satisfiability is considered. We denote their conjunction by $\Ax_\varphi$.


\begin{theorem}[relative completeness]
\label{relcomp}
Let $\Ax_\varphi\wedge\varphi$ {\em not} be the negation of a formula that is valid in $\CTL^*$, assuming that the $<_i$-subformulas in $\Ax_\varphi\wedge\varphi$ are treated as atomic propositions. Then $\varphi$ is satisfiable in an extended CGM of the form (\ref{extCGM1}).
\end{theorem}
\oomit{
\begin{proof}
Without loss of generality, $\varphi$ can be assumed to be a state formula. If not, $\exists\varphi$ can be considered instead. Since $\varphi''$ is not the negation of a valid $\CTL^*$ formula, there exists a  CGM $M$ for the extension $AP'$ of $AP$ by the atomic propositions (\ref{auxat}) and $\bm{w}\in R^\fin_M(w_I)$ be such that $M,\bm{w}\models\varphi''$. Let $<_i$ be defined on $R^\fin_M(w_I)$ as described above. Lemmata \ref{correctness} and \ref{partialOrder} entail that this leads to an extended CGM $M'$ of the form (\ref{extCGM1}). A direct check shows that $M'
,\bm{v}\models\chi_1<_i\chi_2\Leftrightarrow p^{<_i}_{\chi_1,\chi_2}$ for all $\bm{v}\in R^\fin_M(w_I)$. Hence $M',\bm{v}\models\varphi$. 
\end{proof}
}

\begin{theorem}[finite systems of objectives]
Assume that $\varphi$ is as in Theorem \ref{relcomp}.
Then $\varphi$ is satisfiable in an extended CGM of the form (\ref{extCGM2}).
\end{theorem}
\oomit{
\begin{proof}
Corollary \ref{relcomp} guarantees the existence of a satisfying CGM $M$.
Lemma \ref{correctness} entails that the defining clause (\ref{prefDef}) for $<_i$ which we use in the construction of $M$ can be specialised to refer only to the satisfaction of $\mathsf{p}^{<_i}_{\chi_1,\chi_2}$ at the root finite run $w_I$:
\begin{equation}\label{prefDef2}
\begin{array}{ll}
\bm{w}_1<_i\bm{w}_2&\mbox{iff there exist some }\chi_1,\chi_2\in\mathrm{X}\mbox{ such that }\\
& M,\bm{w}_k,1\models\chi_k,\ k=1,2,\mbox{ and }\\
& M,w_I\models\mathsf{p}^{<_i}_{\chi_1,\chi_2}.
\end{array}
\end{equation}
Hence, for any two $\bm{w}_1,\bm{w}_2\in R^\infin_M(w_I)$, $\bm{w}_1<_i\bm{w}_2$ can be determined from the unique pair $\chi_1,\chi_2\in\mathrm{X}$ for which $M,\bm{w}_k,1\models\chi_k,\ k=1,2$. This means that we can can define
\[\Theta_{I,i}\defeq\{\{\bm{w}\in R^\infin_M(w_I):M,\bm{w},1\models\chi\}:\chi\in\mathrm{X}\}.\]
\end{proof}
}

\section*{Concluding Remarks} 
We have proposed a way to use temporal logic for strategic reasoning about temporary coalitions, which are outside the immediate scope of the basic constructs of established logics for strategic reasoning such as $\ATL$ and $\SL$ as both $\ATL$'s $\atlD{.}$ and $\SL$'s first order language for strategies as the domain of individuals are meant to model long term individual player strategies. The proposed notation builds on the use of propositional variables to denote decisions and, more generally, strategies, and quantifying over strategies, in $\QCTL^*$. We extended the notation to capture evolving coalition structure too. Furthermore, we have extended $\CTL^*$ with a preference operator on objectives and proposed a complete set of axioms for that operator. We have illustrated the use of the notation by specifying temporary coalition variants of the solution concepts of {\em the core} and {\em dominant strategies}.

\section*{Acknowledgement}

This work was partially supported by Contract DN02/15/19.12.2016 "Space, Time and Modality: Relational, Algebraic and Topological Models" with Bulgarian NSF.

\bibliographystyle{alpha}

\begin{thebibliography}{WvdHW07}

\bibitem[AHK97]{AHK97}
Rajeev Alur, Tom Henzinger, and Orna Kupferman.
\newblock {Alternating-time Temporal Logic}.
\newblock In {\em {Proceedings of FCS'97}}, pages 100--109, 1997.

\bibitem[AHK02]{AHK02}
Rajeev Alur, Tom Henzinger, and Orna Kupferman.
\newblock {Alternating-time temporal logic}.
\newblock {\em {Journal of the ACM}}, 49(5):1--42, 2002.

\bibitem[BBMU12]{DBLP:conf/fossacs/BouyerBMU12}
Patricia Bouyer, Romain Brenguier, Nicolas Markey, and Michael Ummels.
\newblock {Concurrent Games with Ordered Objectives}.
\newblock In {\em {FOSSACS} 2012}, volume 7213 of {\em LNCS}, pages 301--315.
  Springer, 2012.

\bibitem[BBMU15]{DBLP:journals/corr/BouyerBMU15}
Patricia Bouyer, Romain Brenguier, Nicolas Markey, and Michael Ummels.
\newblock {Pure Nash Equilibria in Concurrent Deterministic Games}.
\newblock {\em Logical Methods in Computer Science}, 11(2), 2015.

\bibitem[BLLM09]{BrihayeLLM09}
Thomas Brihaye, Arnaud Da~Costa Lopes, Fran\c{c}ois Laroussinie, and Nicolas
  Markey.
\newblock {{ATL} with Strategy Contexts and Bounded Memory}.
\newblock In {\em LFCS}, volume 5407 of {\em LNCS}, pages 92--106, 2009.

\bibitem[BMS20]{DBLP:journals/tcs/BozzelliMS20}
Laura Bozzelli, Aniello Murano, and Loredana Sorrentino.
\newblock Alternating-time temporal logics with linear past.
\newblock {\em Theor. Comput. Sci.}, 813:199--217, 2020.

\bibitem[BTV07]{DBLP:journals/jancl/BalbianiTV07}
Philippe Balbiani, Tinko Tinchev, and Dimiter Vakarelov.
\newblock {Dynamic Logics of the Region-based Theory of Discrete Spaces}.
\newblock {\em J. Appl. Non Class. Logics}, 17(1):39--61, 2007.

\bibitem[CEW12]{CooperativeGameTheory2012}
Georgios Chalkiadakis, Edith Elkind, and Michael Wooldridge.
\newblock {\em {Computational Aspects of Cooperative Game Theory}}.
\newblock Morgan \& Claypool Publishers, 2012.

\bibitem[FG92]{DBLP:journals/jolli/FingerG92}
Marcelo Finger and Dov~M. Gabbay.
\newblock {Adding a Temporal Dimension to a Logic System}.
\newblock {\em Journal of Logic, Language and Information}, 1(3):203--233,
  1992.

\bibitem[FKL10]{DBLP:conf/tacas/FismanKL10}
Dana Fisman, Orna Kupferman, and Yoad Lustig.
\newblock {Rational Synthesis}.
\newblock In {\em {TACAS} 2010}, volume 6015 of {\em LNCS}, pages 190--204.
  Springer, 2010.

\bibitem[Fre01]{DBLP:conf/ausai/French01}
Tim French.
\newblock {Decidability of Quantifed Propositional Branching Time Logics}.
\newblock In {\em {Australian Joint Conference on Artificial Intelligence}},
  volume 2256 of {\em LNCS}, pages 165--176. Springer, 2001.

\bibitem[Fre06]{French06}
Tim French.
\newblock {\em {Bisimulation Quantifiers for Modal Logics}}.
\newblock {Ph.D. Thesis}, The University of Western Australia, 2006.
\newblock Accessed in 2011 from {\tt
  http://www.csse.uwa.edu.au/\~{}tim/papers/thesis2.pdf}.

\bibitem[Gab87]{Gab89}
Dov~M. Gabbay.
\newblock {The Declarative Past and Imperative Future: Executable Temporal
  Logic for Interactive Systems}.
\newblock In {\em Temporal Logic in Specification, Altrincham, UK, April 8-10,
  1987, Proceedings}, volume 398 of {\em LNCS}, pages 409--448. Springer, 1987.

\bibitem[GD12]{DBLP:conf/clima/GuelevD12}
Dimitar~P. Guelev and Catalin Dima.
\newblock {Epistemic {ATL} with Perfect Recall, Past and Strategy Contexts}.
\newblock In {\em {CLIMA} XIII, 2012}, volume 7486 of {\em LNCS}, pages 77--93.
  Springer, 2012.

\bibitem[GJ04]{GJ04}
Valentin Goranko and Wojtek Jamroga.
\newblock {Comparing Semantics for Logics of Multi-agent Systems}.
\newblock {\em {Synthese}}, 139(2):241--280, 2004.

\bibitem[GS98]{DBLP:journals/igpl/GabbayS98}
Dov~M. Gabbay and Valentin~B. Shehtman.
\newblock {Products of Modal Logics, Part 1}.
\newblock {\em Log. J. {IGPL}}, 6(1):73--146, 1998.

\bibitem[Gue13]{DBLP:journals/corr/abs-1303-0794}
Dimitar~P. Guelev.
\newblock {Reducing Validity in Epistemic {ATL} to Validity in Epistemic
  {CTL}}.
\newblock In {\em {Proceedings 1st International Workshop on Strategic
  Reasoning, {SR}}}, volume 112 of {\em {EPTCS}}, pages 81--89, 2013.

\bibitem[Han04]{HanssonPreferenceLogic}
Sven~Ove Hansson.
\newblock Preference logic.
\newblock In Dov~M. Gabbay and Franz Guenthner, editors, {\em Handbook of
  Philosophical Logic, 2nd Edition}, pages 319--394. Springer, 2004.

\bibitem[IV12]{DBLP:journals/jolli/IvanovV12}
Nikolay Ivanov and Dimiter Vakarelov.
\newblock {A System of Relational Syllogistic Incorporating Full Boolean
  Reasoning}.
\newblock {\em Journal of Logic, Language and Information}, 21(4):433--459,
  2012.

\bibitem[KPV16]{DBLP:journals/amai/KupfermanPV16}
Orna Kupferman, Giuseppe Perelli, and Moshe~Y. Vardi.
\newblock Synthesis with rational environments.
\newblock {\em Ann. Math. Artif. Intell.}, 78(1):3--20, 2016.

\bibitem[LLM12]{DBLP:conf/concur/LopesLM12}
Arnaud Da~Costa Lopes, Fran{\c{c}}ois Laroussinie, and Nicolas Markey.
\newblock {Quantified {CTL:} Expressiveness and Model Checking - (Extended
  Abstract)}.
\newblock In {\em {CONCUR} 2012. Proceedings}, volume 7454 of {\em LNCS}, pages
  177--192. Springer, 2012.

\bibitem[Lor10]{Lorini2010}
Emiliano Lorini.
\newblock {A Logical Account of Social Rationality in Strategic Games}.
\newblock Technical Report IRIT/RT–2010-3–FR, UPS-IRIT, 2010.
\newblock LOFT 9, Accessed from {\tt
  https://www.irit.fr/publis/LILAC/Reports/LoriniGameTheoryRationality.pdf} in
  2020.

\bibitem[MMV10]{DBLP:conf/fsttcs/MogaveroMV10}
Fabio Mogavero, Aniello Murano, and Moshe~Y. Vardi.
\newblock Reasoning about strategies.
\newblock In {\em {FSTTCS} 2010}, volume~8 of {\em LIPIcs}, pages 133--144,
  2010.

\bibitem[Pin07]{DBLP:conf/atva/Pinchinat07}
Sophie Pinchinat.
\newblock {A Generic Constructive Solution for Concurrent Games with Expressive
  Constraints on Strategies}.
\newblock In {\em {ATVA} 2007, Proceedings}, volume 4762 of {\em LNCS}, pages
  253--267. Springer, 2007.

\bibitem[vBvOR05]{van2005preference}
J.~van Benthem, S.~van Otterloo, and O.~Roy.
\newblock {Preference Logic, Conditionals and Solution Concepts in Games}.
\newblock Technical report, Inst. for Logic, Language and Computation, 2005.

\bibitem[vW63]{vonWrightPreference}
G.~H. von Wright.
\newblock {\em {The Logic of Preference}}.
\newblock Edinburgh University Press, Edinburgh, 1963.

\bibitem[WHY11]{WangHY11}
Farn Wang, Chung-Hao Huang, and Fang Yu.
\newblock A temporal logic for the interaction of strategies.
\newblock In {\em CONCUR}, volume 6901 of {\em LNCS}, pages 466--481. Springer,
  2011.

\bibitem[WvdHW07]{WaltherHW07}
Dirk Walther, Wiebe van~der Hoek, and Michael Wooldridge.
\newblock {Alternating-time Temporal Logic with Explicit Strategies}.
\newblock In Dov Samet, editor, {\em TARK}, pages 269--278. ACM Press, 2007.

\end{thebibliography}

\end{document}